\documentclass{llncs}
\usepackage{algorithm}
\usepackage{algpseudocode}
\usepackage{amsmath}
\usepackage{amssymb}
\usepackage{amsfonts}
\usepackage{color}
\usepackage{xspace}
\newcommand{\Prr}[1]{\underset{\ #1}{\Pr}}

\DeclareMathAlphabet\mathbfcal{OMS}{cmsy}{b}{n}

\newcommand{\ignore}[1]{}
\renewcommand{\Pr}{{\bf Pr}}
\newcommand{\E}{{\bf E}}

\newcommand{\bQ}{\boldsymbol{Q}}
\newcommand{\bcQ}{\mathbfcal{Q}}
\newcommand{\bphi}{\boldsymbol{\phi}}

\newcommand{\bT}{{\mathbf T}}

\begin{document}
\title{Improved Lower Bound for Estimating the Number of Defective Items}
\author{{\bf Nader H. Bshouty} \\ Dept. of Computer Science\\ Technion,  Haifa\\ bshouty@cs.technion.ac.il}
\institute{}
\maketitle

\begin{abstract}
Let $X$ be a set of items of size $n$ that contains some defective items, denoted by $I$, where $I \subseteq X$. In group testing, a {\it test} refers to a subset of items $Q \subset X$. The outcome of a test is $1$ if $Q$ contains at least one defective item, i.e., $Q\cap I \neq \emptyset$, and $0$ otherwise.

We give a novel approach to obtaining lower bounds in non-adaptive randomized group testing. The technique produced lower bounds that are within a factor of $1/{\log\log\stackrel{k}{\cdots}\log n}$ of the existing upper bounds for any constant~$k$. Employing this new method, we can prove the following result.

For any fixed constants $k$, any non-adaptive randomized algorithm that, for any set of defective items $I$, with probability at least $2/3$, returns an estimate of the number of defective items $|I|$ to within a constant factor requires at least 
$$\Omega\left(\frac{\log n}{\log\log\stackrel{k}{\cdots}\log n}\right)$$ tests. 

Our result almost matches the upper bound of $O(\log n)$ and solves the open problem posed by Damaschke and Sheikh Muhammad in~\cite{DamaschkeM10,DamaschkeM10b}. Additionally, it improves upon the lower bound of $\Omega(\log n/\log\log n)$ previously established by Bshouty~\cite{Bshouty19}.
\end{abstract}

\section{Introduction}
Let $X$ be a set of $n$ items, among which are defective items denoted by $I \subseteq X$. In the context of group testing, a {\it test} is a subset $Q\subseteq X$ of items, and its result is $1$ if $Q$ contains at least one defective item (i.e., $Q\cap I\not=\emptyset$), and $0$ otherwise.

Although initially devised as a cost-effective way to conduct mass blood testing \cite{D43}, group testing has since been shown to have a broad range of applications. These include DNA library screening \cite{ND00}, quality control in product testing \cite{SG59}, file searching in storage systems \cite{KS64}, sequential screening of experimental variables \cite{L62}, efficient contention resolution algorithms for multiple-access communication \cite{KS64,W85}, data compression \cite{HL02}, and computation in the data stream model \cite{CM05}. Additional information about the history and diverse uses of group testing can be found in \cite{Ci13,DH00,DH06,H72,MP04,ND00} and their respective references.

{\it Adaptive} algorithms in group testing employ tests that rely on the outcomes of previous tests, whereas {\it non-adaptive} algorithms use tests independent of one another, allowing all tests to be conducted simultaneously in a single step. Non-adaptive algorithms are often preferred in various group testing applications~\cite{DH00,DH06}.

In this paper, we give a novel approach to obtaining lower bounds in non-adaptive group testing. The technique produced lower bounds that are within a factor of $1/{\log\log\stackrel{k}{\ldots}\log n}$ of the existing upper bounds for any constant $k$. Employing this new method, we can prove a new lower bound for non-adaptive estimation of the number of defective items.

Estimating the number of defective items $d:=|I|$ to within a constant factor of $\alpha$ is the problem of identifying an integer $D$ that satisfies $d\le D\le \alpha d$. This problem is widely utilized in a variety of applications~\cite{ChenS90,Swallow85,Thompson62,WalterHB80,GatwirthH89}. 

Estimating the number of defective items in a set $X$ has been extensively studied, with previous works including~\cite{BB18,ChengX14,DamaschkeM10,DamaschkeM10b,FalahatgarJOPS16,RonT14}. 
In this paper, we focus specifically on studying this problem in the non-adaptive setting. Bshouty~\cite{Bshouty19} showed that deterministic algorithms require at least $\Omega(n)$ tests to solve this problem. For randomized algorithms,
Damaschke and Sheikh Muhammad ~\cite{DamaschkeM10b} presented a non-adaptive randomized algorithm that makes $O(\log n)$ tests and, with high probability,
returns an integer $D$ such that $D\ge d$ and $\E[D]=O(d)$. Bshouty \cite{Bshouty19} proposed a polynomial time randomized algorithm that makes $O(\log n)$ tests and, with probability at least $2/3$, returns an estimate of the number of defective items within a constant factor. Damaschke and Sheikh Muhammad \cite{DamaschkeM10b} gave the lower bound of $\Omega(\log n)$; however, this result holds only for algorithms that select each item in each test uniformly and independently with some fixed probability. They conjectured that any randomized algorithm with a constant failure probability also requires $\Omega(\log n)$ tests. Bshouty \cite{Bshouty19} proves this conjecture up to a factor of $\log\log n$. In this paper, we establish a lower bound of  $$\Omega\left(\frac{\log n}{(c\log^* n)^{(\log^*n)+1}}\right)$$ tests, where $c$ is a constant and $\log^*n$ is the smallest integer $k$ such that $\log\log \stackrel{k}{\ldots}\log n<2$.
It follows that the lower bound is
$$\Omega\left(\frac{\log n}{\log\log \stackrel{k}{\ldots}\log n}\right)$$
for any constant $k$.

An attempt was made to establish this bound in~\cite{Bshouty18}; however, an error was discovered in the proof. As a result, the weaker bound of Bshouty, $\Omega(\log n/\log\log n)$, was proved and published in~\cite{Bshouty19}.

The paper is organized as follows: The next subsection introduces the technique used to prove the lower bound. Section 2 defines the notation and terminology used throughout the paper. Section 3 contains the proof of the lower bound.

\subsection{Old and New Techniques}

In this section, we will explain both the old and new techniques used to prove the lower bounds.

Let $X=[n]$ be the set of items, and let ${\cal I}=2^{X}$ be the collection of all the possible sets of defective items. The objective is to establish a lower bound for the test complexity of any non-adaptive randomized algorithm that, for any set of defective items $I\in {\cal I}$, with probability at least $2/3$, returns an integer ${\cal P}(I)$ such that\footnote{The constant $2$ can be replaced by any constant.} $|I|\le {\cal P}(I)\le 2|I|$.

\subsubsection{Old Technique:}
The method used by Bshouty in~\cite{Bshouty19} can be described as follows. Suppose we have a non-adaptive randomized algorithm ${\cal A}$ that makes $s$ tests, denoted by the random variable set ${\cal Q}=\{Q_1, \ldots, Q_s\}$. For every set of defective items $I\in {\cal I}$, the algorithm returns an integer ${\cal P}(I)$ such that, with probability at least $2/3$,  $|I|\le {\cal P}(I)\le 2|I|$. 

First, he defines a partition of the set of tests ${\cal Q} = \bigcup_{i=1}^r {\cal Q}^{(i)}$, where each ${\cal Q}^{(i)}$, $i\in [r]$, contains the set of tests of sizes in the interval $[n_i, n_{i+1}]$, with $n_0=1$ and $n_{i+1} = \text{poly}(\log n) \cdot n_i$.  There are $r=\Theta(\log n/\log\log n)$ such intervals. Let $c$ be a small constant. By Markov's bound, there exists $j$ (that depends only on ${\cal A}$, not the seed of ${\cal A}$) such that, with high probability (at least $1-1/c$),  $|{\cal Q}^{(j)}|\le c s/r$. 

He then identifies an integer $d$ that depends on $j$ (and therefore on ${\cal A}$) such that, for every $m\in [d,4d]$ and for {\it a uniform random} $I\in {\cal I}_m:=\{I\in {\cal I}: |I|=m\}$, the outcomes of all tests that lie outside ${\cal Q}^{(j)}$ can be determined (without having to perform a test) with high probability. This probability is obtained by applying the union bound to the probability that the answer to each test in ${\cal Q}\backslash {\cal Q}^{(j)}$ can be determined by a randomly selected $I\in {\cal I}':=\cup_{m\in [d,4d]}I_m$. The key idea here is that since ${\cal Q}^{(j)}=\{ q\in{\cal Q}:|Q|\in [n_j,poly(\log n)n_j]\}$, there is $d$ such that for a random uniform set of a defective item of size $m\in [d,4d]$, with high probability, the answers to all the tests $Q$ that satisfy $|Q|>poly(\log n)n_j$ are $1$, and, with high probability, the answers to all the tests $Q$ that satisfy $|Q|<n_j$ are $0$. 

This proves that the set of tests ${\cal Q}^{(j)}$ can, with high probability, estimate the size of a set of defective items of a uniformly random $I\in {\cal I}'$, i.e., $|I|\in [d,4d]$. In particular, it can, with high probability, distinguish\footnote{This is because the algorithm for $|I|=d$ return an integer in the interval $[d,2d]$ and for $|I|=4d$ returns an integer in the interval $[4d,8d]$ and both intervals are disjoint.} between a set of defective items of size $d$ and a set of defective items of size $4d$. 

If $c s/r<1$, then $|{\cal Q}^{(j)}|<1$, and therefore, with high probability, ${\cal Q}^{(j)}=\emptyset$. 
This leads to a contradiction. This is because the algorithm cannot, with high probability, distinguish between the case where $|I|=d$ and $|I|=4d$ without any tests. Therefore, $c s/r\ge 1$, resulting in the lower bound $s> r/c=\Omega(\log n/\log\log n)$ for any non-adaptive randomized algorithm that solves the estimation problem. 

\subsubsection{New Technique:} The union bound required for proving that the outcome of the tests in ${\cal Q}\backslash{\cal Q}^{(j)}$ can be determined with high probability necessitates a small enough value of $r$. Consequently, to satisfy the condition $c s/r<1$, $s$ must also be sufficiently small. This is the bottleneck in getting a better lower bound.

We surmount the bottleneck in this paper by implementing the following technique. Let $\tau=\log^*n$. As in the old technique, we define a partition of the set of tests ${\cal Q}=\cup_{i=1}^r{\cal Q}^{(i)}$.  By Markov's bound there is $j$ such that, with probability at least $1-1/\tau$,  $|{\cal Q}^{(j)}|\le \tau s/r$. Next, we identify a subset ${\cal I}'\subset {\cal I}$ such that, for a uniform random $I\in {\cal I}'$, the outcomes of all the tests that lie outside ${\cal Q}^{(j)}$ can be determined with a probability of at least $1-1/\tau$. 

We then present the following algorithm ${\cal A}'$ that, with high probability, solves the problem for defective sets $I\in {\cal I}'$: 

\begin{algorithm}[H]
\caption{Algorithm ${\cal A}'$ for solving problem ${\cal P}$}
\begin{algorithmic}[1]
\Require $I\in {\cal I}'$
\Ensure An estimate ${D}$ of $|I|$
\State Let $\phi:X\to X$ be a uniform random permutation.
\State Let ${\cal Q}$ be the set of tests that the algorithm ${\cal A}$ makes.
\State For every test that is outside ${\cal Q}^{(j)}$, return the determined answer.
\State If $|{\cal Q}^{(j)}|\le \tau s/r$, then make the tests in ${\cal Q}^{(j)}_\phi:=\{\phi(Q)|Q\in {\cal Q}^{(j)}\}$. Otherwise, FAIL.
\State Run the algorithm ${\cal A}$ with the above answers to get an estimation $D$ of $|\phi^{-1}(I)|$.
\State Return $D$.
\end{algorithmic}
\end{algorithm}
We then prove that if the algorithm ${\cal A}$ solves the estimation problem for any set of defective items $I$ with a success probability of at least $2/3$, then algorithm~${\cal A}'$ solves the estimation problem for $I\in {\cal I}'$ with a success probability of at least $2/3-2/\tau$. 

This follows from:
\begin{itemize}
\item Making the tests in ${\cal Q}^{(i)}_\phi$ with defective set of item $I$ is equivalent to making the tests in ${\cal Q}$ with defective set of items $\phi^{-1}(I)$, and therefore, with a random uniform defective set of size $|I|$.
\item Since $\phi^{-1}(I)$ is a random uniform set of size $|I|$, the answers to the queries outside ${\cal Q}^{(j)}$ can be determined with high probability.
\item With high probability, $|{\cal Q}^{(j)}|\le \tau s/r$.

\end{itemize}
Now, as before, if $\tau s/r<1$, then, with high probability, ${\cal Q}^{(j)}=\emptyset$, and the above algorithm does not require any tests to be performed. If, in addition, there are two instances $I_1$ and $I_2$ in ${\cal I}'$ where $|I_1|<4|I_2|$ (then, the outcome for $I_1$ cannot be equal to the outcome of $I_2$), then we get a contradiction. This is because if the algorithm makes no tests, it cannot distinguish between $I_1$ and $I_2$. This contradiction, again, gives the lower bound $r/\tau$ for any non-adaptive randomized algorithm that solves the estimation problem. 

To obtain a better lower bound, we again take the algorithm ${\cal A}'$ that solves ${\cal P}$ for $I\in {\cal I}'$ with the tests $Q':={\cal Q}^{(j)}$ with a success probability $2/3-2/\tau$ and, using the same procedure as before, we generate a new non-adaptive algorithm ${\cal A}''$ that solves ${\cal P}$ for $I\in {\cal I}''\subset {\cal I}'$ using the tests in ${\cal Q}''\subset {\cal Q}'$, with success probability $2/3-4/\tau$. The test complexity of the algorithm ${\cal A}''$ is $\tau^2s/(rr')$ where $r'=O(\log\log n/\log\log\log n)$ is the number of set partitions of ${\cal Q}'$. The lower bound obtained here is now $rr'/\tau^2=\Omega(\log n/((\log^*n)^2\log\log\log n))$, which is better than the lower bound $r/\tau$ obtained before\footnote{We can take $\tau$ as a small constant and get the lower bound $\Omega(\log n/\log\log\log n)$.}.

If this process is repeated $\ell:=\tau/24-\log^*\tau$ times, we get an algorithm that makes $t:=\tau^\ell s/(rr'r''\cdots)$ tests. If $t<1$, the algorithm makes no tests and, with a probability of at least\footnote{Unlike in the previous footnote, $\tau$ cannot be taken as constant here.} $2/3-2(\tau/24)/\tau=7/12>1/2$, it can distinguish between two sets of defective items $I_1$ and $I_2$ that cannot have the same outcome. This gives the lower bound $rr'r''\cdots/(\tau^{\tau/24})=\log n/(c'\tau)^{\tau/24}$ for some constant $c'$.

\subsubsection{Old Attempt}: An attempt was made to establish this bound in~\cite{Bshouty18}; however, an error was discovered in the proof. As a result, the weaker bound of $\Omega(\log n/\log\log n)$ was proved and published in~\cite{Bshouty19}. In~\cite{Bshouty18}, Bshouty did not use Algorithm~1. Instead, he performs the same analysis on ${\cal Q}^{(j)}$ (instead of ${\cal Q}^{(j)}_\phi$), which results in many dependent events in the proof. The key to the success of our analysis is the inclusion of the random permutation $\phi$ in Algorithm 1. This permutation makes the events independent, allowing us to repeat the same analysis for ${\cal Q}^{(j)}$.

\section{Definitions and Notation}
In this section, we introduce some definitions and notation.

We will consider the set of {\it items} $X=[n]=\{1,2,\ldots,n\}$ and the set of {\it defective items} $I\subseteq X$. The algorithm knows $n$ and has access to a test oracle~${\cal O}_I$. The algorithm can use the oracle ${\cal O}_I$ to make a {\it test} $Q\subseteq X$, and the oracle answers ${\cal O}_I(Q):=1$ if $Q\cap I\not=\emptyset$, and ${\cal O}_I(Q):=0$ otherwise. We say that an algorithm $A$, with probability at least $1-\delta$, $\alpha$-{\it estimates} the number of defective items if, for every $I\subseteq X$, $A$ runs in polynomial time in $n$, makes tests with the oracle ${\cal O}_I$, and with probability at least $1-\delta$, returns an integer $D$ such that $|I|\le D\le \alpha |I|$. If $\alpha$ is constant, then we say that the algorithm {\it estimates the number of defective items to within a constant factor}. 

The algorithm is called {\it non-adaptive} if the queries are independent of the answers of previous queries and, therefore, can be executed simultaneously in a single step. Our objective is to develop a non-adaptive algorithm that minimizes the number of tests and provides, with a probability of at least $1-\delta$, an estimation of the number of defective items within a constant factor.

We will denote $\log^{[k]}n=\log\log \stackrel{k}{\ldots}\log n$, $\log^{[0]}n=n$. Notice that $\log\log^{[i]}n=\log^{[i+1]}n$ and $2^{\log^{[i]}n}=\log^{[i-1]}n$. Let ${\mathbb N}=\{0,1,\cdots\}$. For two real numbers $r_1,r_2$, we denote $[r_1,r_2]=\{r\in \mathbb{N}|r_1\le r\le r_2\}$.
Random variables and random sets will be in bold. 

\section{The Lower Bound}

In this section, we prove the lower bound for the number of tests in any non-adaptive randomized algorithm that $\alpha$-estimates the number of defective items for any constant $\alpha$. 

\subsection{Lower Bound for Randomized Algorithm}
In this section, we prove.
\begin{theorem}\label{TH1} Let $\tau=\log^*n$ and $\alpha$ be any constant.
    Any non-adaptive randomized algorithm that, with probability at least $2/3$, $\alpha$-estimates the number of defective items must make at least
    $$\Omega\left(\frac{\log n}{(480\tau)^{\tau+1}}\right)$$
    tests. 
\end{theorem}

We first prove the following.
\begin{lemma}\label{main}
Let $n_1=n$. Let $i\ge 1$ be an integer such that $\log^{[i]}n\ge \tau:=\log^*n$. Suppose there is an integer $n_i=n^{\Omega(1)}\le n$ and a non-adaptive randomized algorithm ${\cal A}_i$ that makes 
\begin{eqnarray}\label{si}
    s_i:=\frac{\log^{[i]} n}{(480\tau)^{\tau-i+2}}
\end{eqnarray} tests and
for every set of defective items $I$ of size $$d\in D_i:=\left[\frac{n}{n_i},\frac{n(\log^{[i-1]}n)^{1/4}}{n_i}\right],$$ with probability at least $1-\delta$, $\alpha$-estimates $d$. Then there is an integer $n_{i+1}=n^{\Omega(1)}\le n$ and a non-adaptive randomized algorithm ${\cal A}_{i+1}$ that makes $s_{i+1}$ tests and for every set of defective items $I$ of size $d\in D_{i+1}$, with probability at least $1-\delta-1/(12\tau)$, $\alpha$-estimates $d$. 
\end{lemma}
\begin{proof} 
Let $$N_{i}=\left[\frac{n_i}{(\log^{[i-1]}n)^{1/4}},n_i\right].$$
We will be interested in all the queries $Q$ of the algorithm ${\cal A}_i$ that satisfies $|Q|\in N_i$. We now partition $N_i$ into smaller sets.  
Let $$N_{i,j}=\left[\frac{n_i}{(\log^{[i]}n)^{4j+4}},\frac{n_i}{(\log^{[i]}n)^{4j}}\right]$$
where $j=[0,r_i-1]$ and
\begin{eqnarray}\label{ri}
    r_i=\frac{\log^{[i]}n}{16\log^{[i+1]}n}.
\end{eqnarray}
Since the lowest endpoint of the interval $N_{i,r_i-1}$ is
$$\frac{n_i}{(\log^{[i]}n)^{4(r_i-1)+4}}=\frac{n_i}{2^{4r_i\log^{[i+1]}n}}=\frac{n_i}{2^{(1/4)\log^{[i]}n}}=\frac{n_i}{(\log^{[i-1]}n)^{1/4}}$$
and the right endpoint of $N_{i,0}$ is $n_i$, we have, $N_i=\cup_{j=0}^{r_i-1}N_{i,j}$. 

Let ${\mathbfcal{Q}}=\{\bQ_1,\ldots,\bQ_{s_i}\}$ be the (random variable) tests that the randomized algorithm ${\cal A}_i$ makes. Let $\bT_j$ be a random variable that represents the number of tests $\bQ\in \mathbfcal{Q}$ that satisfies $|\bQ|\in N_{i,j}$. Since ${\cal A}_i$ makes $s_i$ tests, we have $\bT_0+\cdots+\bT_{r_i}\le s_i$. Therefore, by (\ref{si}) and (\ref{ri}), (in the expectation, $\E_j$, $j$ is uniformly at random over $[0,r_i-1]$ and the other $\E$ is over the random seed of the algorithm~${\cal A}_i$)
$$\E_j\left[\E[\bT_j]\right]=\E\left[\E_j[\bT_j]\right]\le \frac{s_i}{r_i}=\frac{16 \log^{[i+1]}n}{(480\tau)^{\tau-i+2}}.$$
Therefore, there is $0\le j_i\le r_i-1$ that depends only on the algorithm ${\cal A}_i$ (not the seed of the algorithm) such that 
$$\E[\bT_{j_i}]\le \frac{16 \log^{[i+1]}n}{(480\tau)^{\tau-i+2}}.$$
By Markov's bound, with probability at least $1-16/(480\tau)=1-1/(30\tau)$,
\begin{eqnarray}\label{sip1}
    |\{{\bf Q}\in \mathbfcal{Q}:|{\bf Q}|\in N_{i,j_i}\}|=\bT_{j_i}\le \frac{\log^{[i+1]}n}{(480\tau)^{\tau-i+1}}=s_{i+1}.
\end{eqnarray}
Define
\begin{eqnarray}\label{nip1}
    n_{i+1}=\frac{n_i}{(\log^{[i]}n)^{4j_i+2}}.
\end{eqnarray}
Since $n_i=n^{\Omega(1)}$ and
$$(\log^{[i]}n)^{4j_i+2}\le (\log^{[i]}n)^{4r_i-2}=\frac{(\log^{[i-1]}n)^{1/4}}{(\log^{[i]}n)^2},$$
we have that $n_{i+1}=n^{\Omega(1)}\le n$. Notice that this holds even for $i=1$. This is because $n_1=n$ and $(\log^{[0]}n)^{1/4}=n^{1/4}$ so $n_2\ge n^{1/4}\log^2n=n^{\Omega(1)}$.

Consider the following randomized algorithm ${\cal A}_{i}'$:
\begin{enumerate}
    \item Let $\bcQ=\{\bQ_1,\ldots,\bQ_{s_i}\}$ be the set of tests of ${\cal A}_i$.
    \item Choose a uniformly at random permutation $\bphi:[n]\to [n]$.
    \item Let $\bcQ'=\{\bQ_1',\ldots,\bQ_{s_i}'\}$ where $\bQ_i'=\bphi(\bQ_i):=\{\bphi(q)|q\in \bQ_i\}$.
    \item Make all the tests in $\bcQ'$ and give the answer to ${\cal A}_i$.
    \item Run ${\cal A}_i$ with the above answers on $\bcQ'$ and output what ${\cal A}_i$ outputs.
\end{enumerate}

Consider the following algorithm ${\cal A}_{i+1}$:
\begin{enumerate}
    \item Let $\bcQ=\{\bQ_1,\ldots,\bQ_{s_i}\}$ be the set of tests of ${\cal A}_i$.
    \item Choose a uniformly at random permutation $\bphi:[n]\to [n]$.
    \item Let $\bcQ'=\{\bQ_1',\ldots,\bQ_{s_i}'\}$ where $\bQ_i'=\bphi(\bQ_i):=\{\phi(q)|q\in \bQ_i\}$.
    \item For all the tests in $$\bcQ_0:=\left\{\bQ_i'\in \bcQ':|\bQ_i'|\le \frac{n_i}{(\log^{[i]}n)^{4j_i+4}}\right\},$$
    answer $0$.
    \item For all the tests in $$\bcQ_1:=\left\{\bQ_i'\in \bcQ':|\bQ_i'|\ge \frac{n_i}{(\log^{[i]}n)^{4j_i}}\right\},$$
    answer $1$.
    \item Let \begin{eqnarray*}\bcQ''&=&\left\{\bQ_i'\in \bcQ':\frac{n_i}{(\log^{[i]}n)^{4j_i+4}}<|\bQ_i'|< \frac{n_i}{(\log^{[i]}n)^{4j_i}}\right\}\\
    &=&\{{\bf Q}_i'\in \mathbfcal{Q}':|{\bf Q}_i'|\in N_{i,j_i}\}\end{eqnarray*}
    \item If $$|\bcQ''|>s_{i+1}=\frac{\log^{[i+1]}n}{(480\tau)^{\tau-i+1}}$$ return $-1$ (FAIL) and halt.
    \item Make all the tests in $\bcQ''$.
    \item Run ${\cal A}_i$ with the above answers on $\bcQ'$ and output what ${\cal A}_i$ outputs. 
\end{enumerate}

We now show that for every set of defective items $|I|$ of size
\begin{eqnarray}\label{dee}
d\in D_{i+1}:=\left[\frac{n}{n_{i+1}},\frac{n(\log^{[i]}n)^{1/4}}{n_{i+1}}\right],
\end{eqnarray}
with probability at least $1-\delta-1/(12\tau)$, algorithm ${\cal A}_{i+1}$ $\alpha$-estimates $d$ using $s_{i+1}$ tests. 

In algorithm ${\cal A}_{i+1}$, Step 8 is the only step that makes tests. Therefore, by step~7 the test complexity of ${\cal A}_{i+1}$ is $s_{i+1}$. 

By the definition of $D_i$, and since, by (\ref{nip1}), $n/n_{i+1}>n/n_i$, and, by (\ref{nip1}) and~(\ref{ri})
$$\frac{n(\log^{[i]}n)^{1/4}}{n_{i+1}}=\frac{n}{n_i}(\log^{[i]}n)^{4j_i+2.25}\le \frac{n}{n_i}(\log^{[i]}n)^{4r_i-1.75}\le \frac{n(\log^{[i-1]}n)^{1/4}}{n_i},$$  we can conclude that $D_{i+1}\subset D_i$.

Consider the following events:
\begin{enumerate}
    \item Event $M_0$: For some $\bQ'=\bphi(\bQ)\in \bcQ'$ such that $|\bQ'|\le n_i/(\log^{[i]}n)^{4j_i+4}$ (i.e., $\bQ'\in \bcQ_0$),  $\bQ'\cap I\not=\emptyset$ (i.e., the answer to the test $\bQ$ is $1$).
    \item Event $M_1$: For some $\bQ'=\bphi(\bQ)\in \bcQ'$ such that $|\bQ'|\ge n_i/(\log^{[i]}n)^{4j_i}$, (i.e., $\bQ'\in \bcQ_1$) $\bQ'\cap I=\emptyset$ (i.e., the answer to the test $\bQ$ is $0$).
    \item Event $W$: $|\bcQ''|>s_{i+1}=\frac{\log^{[i+1]}n}{(480\tau)^{\tau-i+1}}.$
\end{enumerate}
The success probability of the algorithm ${\cal A}_{i+1}$ on a set $I$ of defective items with $|I|=d\in D_{i+1}$ is (here the probability is over $\boldsymbol{\phi}$ and the random tests $\mathbfcal{Q}$) 
\begin{eqnarray*}
    \Pr[{\cal A}_{i+1}\mbox{\ succeeds on\ }I]&=& \Pr[({\cal A}_{i}'\mbox{\ succeeds on\ }I)\wedge \bar M_0\wedge \bar M_1\wedge \bar W]\\
    &\ge & \Pr[{\cal A}_i'\mbox{\ succeeds on\ }I]-\Pr[M_0\vee M_1\vee W]\\
    &\ge & \Pr[{\cal A}_i'\mbox{\ succeeds on\ }I]-\Pr[M_0]-\Pr[M_1]-\Pr[W].
\end{eqnarray*}

Now, since $|\bphi^{-1}(I)|=|I|$ and $\bQ'\cap I=\bphi(\bQ)\cap I\not=\emptyset$ if and only if $\bQ\cap \bphi^{-1}(I)\not=\emptyset$,
\begin{eqnarray*}
\Prr{\boldsymbol{\phi},\mathbfcal{Q}}[{\cal A}_i'\mbox{\ succeeds on\ }I]=\Prr{\boldsymbol{\phi},\mathbfcal{Q}}[{\cal A}_i\mbox{\ succeeds on\ }\phi^{-1}(I)]\ge 1-\delta.
\end{eqnarray*}
Therefore, to get the result, it is enough to show that $\Pr[M_0]\le 1/(300\tau), $ $\Pr[M_1]\le 1/(300\tau)$ and $\Pr[W]\le 1/(30\tau)$.

First, since $|\bQ'|=|\bphi(\bQ)|=|\bQ|$ we have
$$|\bcQ''|=|\{\bQ_i':|\bQ_i'|\in N_{i,j_i}\}|=|\{\bQ_i:|\bQ_i|\in N_{i,j_i}\}|={\bf T}_{j_i}.$$
By (\ref{sip1}), with probability at most $1/(30\tau)$, $$|\bcQ''|={\bf T}_{j_i}>\frac{\log^{[i+1]}n}{(480\tau)^{\tau-i+1}}.$$ 
Therefore, $\Pr[W]\le 1/(30\tau)$.

We now will show that $\Pr[M_0]\le 1/(300\tau)$. We have, (A detailed explanation of every step can be found below.)
\begin{eqnarray} \Prr{\boldsymbol{\phi},\mathbfcal{Q}}[M_0]&=&\Prr{\boldsymbol{\phi},\mathbfcal{Q}}[(\exists \boldsymbol{Q}\in \mathbfcal{Q},\boldsymbol{\phi}(\boldsymbol{Q})\in \mathbfcal{Q}_0) \ \boldsymbol{\phi}(\boldsymbol{Q})\cap I\not=\emptyset]\label{up06}\\
&=&\Prr{\boldsymbol{\phi},\mathbfcal{Q}}[(\exists \boldsymbol{Q}\in \mathbfcal{Q},\boldsymbol{\phi}(\boldsymbol{Q})\in \mathbfcal{Q}_0) \ \boldsymbol{Q}\cap\boldsymbol{\phi}^{-1}(I)\not=\emptyset]\label{up07}\\
    &\le& s_i \left(1-\prod_{k=0}^{d-1}\left(1-\frac{n_i}{(\log^{[i]}n)^{4j_i+4}(n-k)}\right)\right)\label{up08}\\
    &\le& s_i \left(1-\left(1-\frac{2n_i}{(\log^{[i]}n)^{4j_i+4}n}\right)^d\right)\label{up09}\\
    &\le& s_i d\frac{2n_i}{(\log^{[i]}n)^{4j_i+4}n}\label{up10}\\
    &\le& \frac{\log^{[i]} n}{(480\tau)^{\tau-i+2}} \cdot \frac{n(\log^{[i]}n)^{1/4}}{n_{i+1}}\frac{2n_i}{(\log^{[i]}n)^{4j_i+4}n} \label{up11}\\
    &=& \frac{\log^{[i]} n}{(480\tau)^{\tau-i+2}} \cdot \frac{n(\log^{[i]}n)^{4j_i+2\frac{1}{4}}}{n_{i}}\frac{2n_i}{(\log^{[i]}n)^{4j_i+4}n} \label{up12}\\
    &=& \frac{2}{(480\tau)^{\tau-i+2}(\log^{[i]}n)^{3/4}}
    \le\frac{1}{300\tau}.\label{up13}
\end{eqnarray}
(\ref{up06}) follows from the definition of the event $M_0$. (\ref{up07}) follows from the fact that for any permutation $\phi:[n]\to[n]$ and two sets $X,Y\subseteq [n]$, $\phi(X)\cap Y\not =\emptyset$ is equivalent to $X\cap \Phi^{-1}(Y)\not=\emptyset$. (\ref{up08}) follows from: 
\begin{enumerate}
    \item The union bound and $|\bcQ_0|\le |\bcQ|= s_i$.
    \item For a random at uniform $\boldsymbol{\phi}$, and a $d$-subset of $[n]$, $\boldsymbol{\phi}^{-1}(I)$ is a random uniform $d$-subset of $[n]$.
    \item For every $\bQ\in \bcQ$ such that $\phi(\bQ)\in \bcQ_0$, $|\bQ|=|\phi(\bQ)|\le {n_i}/{(\log^{[i]}n)^{4j_i+4}}$.
\end{enumerate}
(\ref{up09}) follows from the fact that since $d\in D_{i+1}$ and $n_{i+1}=n^{\Omega(1)}$ by (\ref{dee}), we have $d\le n/2$. (\ref{up10}) follows from the inequality $(1-x)^d\ge 1-dx$. (\ref{up11}) follows from (\ref{si}) and (\ref{dee}). (\ref{up12}) follows from  (\ref{nip1}). (\ref{up13}) follows from the fact that since $\log^{[i]}n\ge \tau=\log^*n$, we have $i\le\tau$, and therefore $(480\tau)^{\tau-i+2}\ge (480\tau)^{2}\ge 600\tau$.

We now prove that $\Pr[M_1]\le 1/(300\tau)$. 
\begin{eqnarray}    \Prr{\boldsymbol{\phi},\mathbfcal{Q}}[M_1]&=&\Prr{\boldsymbol{\phi},\mathbfcal{Q}}[(\exists \boldsymbol{Q}\in \mathbfcal{Q},\boldsymbol{\phi}(\boldsymbol{Q})\in \mathbfcal{Q}_1) \ \boldsymbol{\phi}(\boldsymbol{Q})\cap I=\emptyset]\label{lo14}\\
&=&\Prr{\boldsymbol{\phi},\mathbfcal{Q}}[(\exists \boldsymbol{Q}\in \mathbfcal{Q},\boldsymbol{\phi}(\boldsymbol{Q})\in \mathbfcal{Q}_1) \ \boldsymbol{Q}\cap\boldsymbol{\phi}^{-1}(I)=\emptyset]\label{lo15}\\
    &\le& s_i \prod_{k=0}^{d-1}\left(1-\frac{n_i}{(\log^{[i]} n)^{4j_i}(n-k)}\right)\label{lo16}\\
    &\le& s_i \left(1-\frac{n_i}{(\log^{[i]} n)^{4j_i}n}\right)^d\le s_i\exp\left({-\frac{dn_i}{(\log^{[i]} n)^{4j_i}n}}\right)\label{lo17}\\
    &\le &\frac{\log^{[i]} n}{(480\tau)^{\tau-i+2}}\exp\left(-\frac{\frac{n}{n_{i+1}}\frac{n_i}{(\log^{[i]}n)^{4j_i}}}{n}\right)\ \label{lo18}\\
    &\le &\frac{\log^{[i]} n}{(480\tau)^{\tau-i+2}}\exp(-(\log^{[i]}n)^2)\label{lo19}\\
    &\le& \frac{1}{300\tau}.\label{lo20}
\end{eqnarray}
(\ref{lo14}) follows from the definition of the event $M_1$. (\ref{lo15}) follows from the fact that for any permutation $\phi:[n]\to[n]$ and two sets $X,Y\subseteq [n]$, $\phi(X)\cap Y =\emptyset$ is equivalent to $X\cap \Phi^{-1}(Y)=\emptyset$. In (\ref{lo16}), we again use the union bound, $|\bcQ_1|\le |\bcQ|= s_i$, the fact that $\phi^{-1}(I)$ is a random uniform $d$-subset, and for $\bQ'\in \bcQ_1$, $|\bQ'|\ge n_i/(\log^{[i]}n)^{4j_i}$. (\ref{lo17}) follows from the inequalities $(1-y/(n-k))\le (1-y/n)$ and $1-x\le e^{-x}$ for every $x$, and $y\ge 0$. (\ref{lo18}) follows from (\ref{si}) and (\ref{dee}). (\ref{lo19}) follows from (\ref{nip1}). (\ref{lo20}) follows from the fact that $i\le\tau$, and therefore $(480\tau)^{\tau-i+2}\ge (480\tau)^{2}\ge 300\tau$.\qed
\end{proof}

We are now ready to prove Theorem~\ref{TH1}.
\begin{proof}
Suppose, for the contrary, there is a non-adaptive randomized algorithm ${\cal A}_1$ that, with probability at least $2/3$, $\alpha$-estimates the number of defective items and makes
$$m:=\frac{\log n}{(480\tau)^{\tau+1}}$$
tests. Recall that $n_1=n$ and $\log^{[0]}n=n$. We use Lemma~\ref{main} with $\delta=1/3$, $D_1=[1,n^{1/4}]$ and $s_1=m$. 

Now let $\ell$ be an integer such that $\log\log^*n<\log^{[\ell]}n\le \log^*n=\tau$. Then $\log^{[\ell-1]}n>2^{\log^{[\ell]}n}>2^{\log\log^* n}=\tau$. Now use Lemma~\ref{main} with $i=\ell-1$ and get
$$s_\ell=\frac{\log^{[\ell]}n}{(480\tau)^{\tau-\ell+2}}\le \frac{\tau}{(480\tau)^2}<1.$$

So, algorithm ${\cal A}_\ell$ makes no tests and with probability at least $2/3-\ell/(12\tau)\ge 7/12>1/2$ $\alpha$-estimates the size of defective items $I$ provided that
$$|I|\in D_\ell=\left[\frac{n}{n_\ell},\frac{n(\log^{[\ell-1]}n)^{1/4}}{n_\ell}\right].$$
In particular, with probability more than $1/2$, we can distinguish between defective sets of size $n/n_\ell$ and of size greater than $\alpha n/n_\ell$ without performing any test, which is impossible. 
A contradiction.\qed
\end{proof}

\section{Conclusion}
In this paper, we have presented a novel approach to obtaining lower bounds in non-adaptive randomized group testing. Our technique has allowed us to establish a lower bound of $\Omega(\log n/((c\log^*n)^{\log^*n})$, for some constant $c$, on the test complexity of any randomized non-adaptive algorithm that estimates the number of defective items within a constant factor. This lower bound significantly improves upon the previous bound of $\Omega(\log n/\log\log n)$ that was established in~\cite{Bshouty19}.

The key to our success was the introduction of a random permutation $\phi$ in Algorithm 1, which enabled us to make the events independent and repeat the analysis to the set ${\cal Q}^{(j)}$. This crucial step allowed us to overcome the bottleneck in the previous technique and achieve a better lower bound.

An open problem is to establish a lower bound of $\Omega(\log n)$ for the test complexity in non-adaptive randomized algorithms that estimate the number of defective items within a constant factor.

\bibliography{LogNLB}
\bibliographystyle{plain}

\end{document}